\documentclass[runningheads]{llncs}

\usepackage{tikz}
\usetikzlibrary{shapes,decorations,arrows,shadows,mindmap,trees}

\usepackage{wrapfig}

\usepackage{amssymb,amsmath}

\usepackage{subcaption}
\begin{document}

\title{An Improved Scheme in the Two Query Adaptive Bitprobe Model}

% If the paper title is too long for the running head, you can set
% an abbreviated paper title here
\titlerunning{Two Query Four Element Bitprobe Scheme}

\author{Mirza Galib Anwarul Husain Baig \and
Deepanjan Kesh \and Chirag Sodani}

% First names are abbreviated in the running head.
% If there are more than two authors, 'et al.' is used.
%\authorrunning{F. Author et al.}

\institute{Indian Institute of Technology Guwahati, Guwahati, Assam 781039, India \email{\{mirza.baig,deepkesh,chirag.sodani\}@iitg.ac.in}}

\maketitle

\begin{abstract}

In this paper, we look into the adaptive bitprobe model that stores subsets of size at most four from a universe of size $m$, and answers membership queries using two bitprobes. We propose a scheme that stores arbitrary subsets of size four using $\mathcal{O}(m^{5/6})$ amount of space. This improves upon the non-explicit scheme proposed by Garg and Radhakrishnan~\cite{DBLP:conf/soda/GargR15} which uses $\mathcal{O}(m^{16/17})$ amount of space, and the explicit scheme proposed by Garg~\cite{garg15} which uses $\mathcal{O}(m^{14/15})$ amount of space.

\end{abstract}

\section{Introduction}

Consider the following static membership problem -- given a universe $\mathcal{U}$ containing $m$ elements, we want to store an arbitrary subset $\mathcal{S}$ of $\mathcal{U}$ whose size is at most $n$, such that we can answer membership queries of the form ``Is $x$ in $\mathcal{S}$?'' Solutions to problems of this nature are called {\em schemes} in the literature. The resources that are considered to evaluate the schemes are the size of the data structure devised to store the subset $\mathcal{S}$, and the number of bits read of the data structure to answer the membership queries, called {\em bitprobes}. The notations for the space used and the number of bitprobes required are $s$ and $t$, respectively. This model of the static membership problem is called the {\em bitprobe model}.

Schemes in the bitprobe model are classified as {\em adaptive} and {\em non-adaptive}. If the location where the current bitprobe is going to be depends on the answers obtained from the previous bitprobes, then such schemes are called {\em adaptive schemes}. On the other hand, if the location of the current bitprobe is independent of the answers obtained in the previous bitprobes, then such schemes are called {\em non-adaptive schemes}. Radhakrishnan {\em et al.}~\cite{DBLP:conf/esa/RadhakrishnanRR01} introduced the notation $(n,m,s,t)_A$ and $(n,m,s,t)_N$ to denote the adaptive and non-adaptive schemes, respectively. Sometimes the space requirement of the two classes of schemes will also be denoted as $s_A(n,m,t)$ and $s_N(n,m,t)$, respectively.

\subsection{The Bitprobe Model}

		\begin{figure}[t]
			\centering
			\begin{tikzpicture}
				\node[ellipse, fill=gray!30] (A) at (5,4) {$\mathcal{A}$};
				\node[ellipse, fill=gray!30] (B) at (2,2) {$\mathcal{B}$};
				\node[ellipse, fill=gray!30] (C) at (8,2) {$\mathcal{C}$};
				\node[ellipse, fill=gray!30] (D) at (0,0) {{\tt No}};
				\node[ellipse, fill=gray!30] (E) at (4,0) {{\tt Yes}};
				\node[ellipse, fill=gray!30] (F) at (6,0) {{\tt No}};
				\node[ellipse, fill=gray!30] (G) at (10,0) {{\tt Yes}};
	
				\draw[line width=1pt, ->] (A) -- (B) node[midway, above] {$0$};
				\draw[line width=1pt, ->] (A) -- (C) node[midway, above] {$1$};
	
				\draw[line width=1pt, ->] (B) -- (D) node[midway, above] {$0$};
				\draw[line width=1pt, ->] (B) -- (E) node[midway, above] {$1$};
		
				\draw[line width=1pt, ->] (C) -- (F) node[midway, above] {$0$};
				\draw[line width=1pt, ->] (C) -- (G) node[midway, above] {$1$};
			\end{tikzpicture}
			\caption{The decision tree of an element.}
			\label{fig:tree}
		\end{figure}
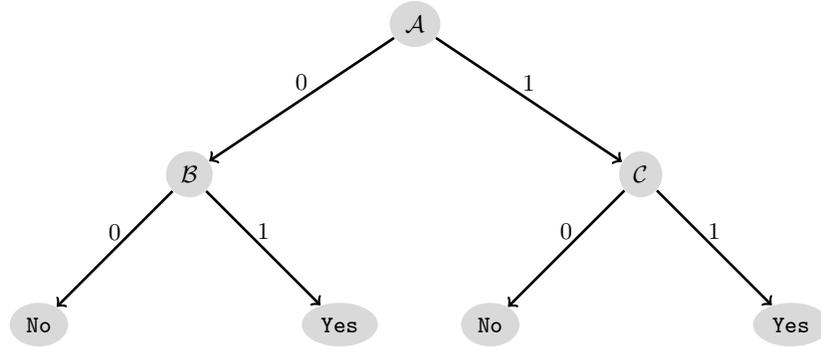

The scheme presented in this paper is an adaptive scheme that uses two bitprobes to answer membership queries. We now discuss in detail the bitprobe model in the context of two adaptive bitprobes.

The data structure in this model consists of three tables -- $\mathcal{A}, \mathcal{B}$, and $\mathcal{C}$ -- arranged as shown in Figure~\ref{fig:tree}. Any element $e$ in the universe $\mathcal{U}$ has a location in each of these three tables, which are denoted by $\mathcal{A}(e), \mathcal{B}(e)$, and $\mathcal{C}(e)$. By a little abuse of notation, we will use the same symbols to denote the bits stored in those locations.

Any bitprobe scheme has two components -- the {\em storage} scheme, and the {\em query} scheme.  Given a subset $\mathcal{S}$, the storage scheme sets the bits in the three tables such that the membership queries can be answered correctly. The flow of the query scheme is traditionally captured in a tree structure, called the {\em decision tree} of the scheme (Figure~\ref{fig:tree}). It works as follows. Given a query ``Is $x$ in $\mathcal{S}$?'', the first bitprobe is made in table $\mathcal{A}$ at location $\mathcal{A}(x)$. If the bit stored is 0, the second query is made in table $\mathcal{B}$, else it is made in table $\mathcal{C}$. If the answer received in the second query is 1, then we declare that the element $x$ is a member of $\mathcal{S}$, otherwise we declare that it is not.

\subsection{The Problem Statement}

As alluded to earlier, we look into adaptive schemes with two bitprobes ($t=2$). When the subset size is one ($n=1$), the problem is well understood -- the space required by the data structure is $\Omega(m^{1/2})$, and we have a scheme that matches this bound~\cite{DBLP:conf/soda/AlonF09,DBLP:conf/esa/LewensteinMNR14}.

For subsets of size two ($n=2$), Radhakrishnan {\em et al.}~\cite{DBLP:conf/esa/RadhakrishnanRR01} proposed a scheme that takes $\mathcal{O}(m^{2/3})$ amount of space, and further conjectured that it is the minimum amount of space required for any scheme. Though progress has been made to prove the conjecture~\cite{DBLP:conf/esa/RadhakrishnanRR01,DBLP:conf/esa/RadhakrishnanSS10}, it as yet remains unproven.

For subsets of size three ($n=3$), Baig and Kesh~\cite{DBLP:conf/walcom/BaigK18} have recently proposed a scheme that takes $\mathcal{O}(m^{2/3})$ amount of space. It has been subsequently proven by Kesh~\cite{kesh32} that $\Omega(m^{2/3})$ is the lower bound for this problem. So, the space complexity question for $n=3$ stands settled.

In this paper, we look into problem where the subset size is four ($n=4$), i.e. an adaptive bitprobe scheme that can store subsets of size atmost four, and answers membership queries using two bitprobes. Garg and Radhakrishnan~\cite{DBLP:conf/soda/GargR15} have proposed a generalised scheme that can store arbitrary subsets of size $n (< \log m)$, and uses $\mathcal{O}(m^{1 - \frac{1}{4n+1}})$ amount of space. For the particular case of $n=4$, the space requirement turns out to be $\mathcal{O}(m^{16/17})$. Garg~\cite{garg15} further improved the bounds to $\mathcal{O}(m^{1 - \frac{1}{4n-1}})$, which improved the scheme for $n=4$ to $\mathcal{O}(m^{14/15})$.

We propose a scheme for the problem whose space requirement is $\mathcal{O}(m^{5/6})$ (Theorem~\ref{thm:final}), thus improving upon the existing schemes in the literature. Our claim is the following:
\[
	s_A(4,m,2) = \mathcal{O}(m^{5/6}). \text{ (Theorem \ref{thm:final})}
\]

\section{Our Data structure}

In this section, we provide a detailed description of our data structure. To achieve a space bound of $o(m)$, more than one element must necessarily share the same location in each of the three tables. We discuss how we arrange the elements of the universe $\mathcal{U}$, and which all elements share the same location in any given table.

Along with the arrangement of elements, we will also talk about the size of our data structure. The next few sections prove the following theorem.

\begin{theorem}
	The size of our data structure is $\mathcal{O}(m^{5/6})$.
	\label{thm:size}
\end{theorem}

\subsection{Table $\mathcal{A}$}

Given the universe $\mathcal{U}$ containing $m$ elements, we partition the universe into sets of size $m^{1/6}$. Borrowing the terminology from Radhakrishnan {\em et al.}~\cite{DBLP:conf/esa/RadhakrishnanRR01}, we will refer to these sets as {\em blocks}. It follows that the total number of blocks in our universe is $m^{5/6}$.

The elements within a block are numbered as $1, 2, 3, \dots, m^{1/6}$. We refer to these numbers as the {\em index} of an element within a block. So, an element of $\mathcal{U}$ can be addressed by the number of the block to which it belongs, and its index within that block.

In table $\mathcal{A}$ of our data structure, we will have one bit
for every block in our universe. As there are $m^{5/6}$ blocks, the
size of table $\mathcal{A}$ is $m^{5/6}$.

\subsection{Superblocks}

\label{sec:superblocks}

The blocks in our universe are partitioned into sets of size $m^{4/6}$. Radhakrishnan {\em et al.}~\cite{DBLP:conf/esa/RadhakrishnanRR01} used the term {\em superblocks} to refer to these sets of blocks, and we will do the same in our discussion. As there are $m^{5/6}$ blocks, the number of superblocks thus formed is $m^{1/6}$. These superblocks are numbered as $1, 2, 3, \dots, m^{1/6}$.

For a given superblock, we arrange the $m^{4/6}$ blocks that it contains into a square grid, whose sides are of size $m^{2/6}$. The blocks of the superblock are placed on the integral points of the grid. The grid is placed at the origin of a two-dimensional coordinate space with its sides parallel to the coordinate axes. This gives a unique coordinate to each of the integral points of the grid, and thus to the blocks placed on those points. It follows that if $(x,y)$ is the coordinate of a point on the grid, then $0 \leq x, y < m^{2/6}$.

We can now have a natural way of addressing the blocks of a given superblock -- we will use the $x$-coordinate and the $y$-coordinate of the point on which the block lies. So, a given block can be uniquely identified by the number of the superblock to which it belongs, and the $x$ and $y$ coordinates of the point on which it lies. Henceforth, we will address any block by a three-tuple of the form $(s, x, y)$, where the $s$ is its superblock number, and $(x,y)$ are the coordinates of the point on which it lies.

To address a particular element of the universe, apart from specifying the block to which it belongs, we need to further state its index within that block. So, an element will be addressed by a four-tuple such as $(s, x, y, i)$, where the first three components specify the block to which it belongs, and the fourth component specifies its index.

\subsection{Table $\mathcal{C}$}

Table $\mathcal{C}$ of our data structure has the space to store one block for every possible point of the grid (described in the previous section). So, for the coordinate $(x,y)$ of the grid, table $\mathcal{C}$ has space to store one block; similarly for all other coordinates. As every superblock has one block with coordinate $(x,y)$, all of these blocks share the same location in table $\mathcal{C}$. So, we can imagine table $\mathcal{C}$ as a square grid containing $m^{4/6}$ points, where each point can store one block.

There are a total of $m^{4/6}$ points in the grid, and the size of a block is $m^{1/6}$, so the space required by table $\mathcal{C}$ is $m^{5/6}$.

\subsection{Lines for Superblocks}

\label{sec:lines}

Given a superblock whose number is $i$, we associate a certain number of lines with this superblock each of whose slopes is $1/i$. In the grid arrangement of the superblock (Section~\ref{sec:superblocks}), we draw enough of these lines of slope $1/i$ so that every grid point falls on one of these lines. Figure~\ref{fig:lines} shows the grid and the lines.

		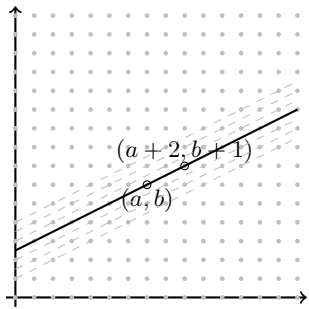
\begin{figure}[t]
			\centering
	
			\begin{tikzpicture}[scale=0.25]
				\draw[thick,->] (-0.5,0) -- (15.5,0);
				\draw[thick,->] (0,-0.5) -- (0,15.5);
				
				\foreach \i in {0,...,15} {
					\foreach \j in {0,...,15} {
						\draw[lightgray,fill=lightgray] (\i,\j) circle (.7ex);
					}
				}
				
				\draw (7,6) circle (1.5ex);
				\node at (7,5.2) {$(a,b)$};
				\draw (9,7) circle (1.5ex);
				\node at (9,7.8) {$(a+2,b+1)$};
				
				\draw[thick,-] (0,2.5) -- (15,10);
				\draw[lightgray,dashed] (0,2) -- (15,9.5);
				\draw[lightgray,dashed] (0,1.5) -- (15,9);
				\draw[lightgray,dashed] (0,1) -- (15,8.5);
				\draw[lightgray,dashed] (0,3) -- (15,10.5);
				\draw[lightgray,dashed] (0,3.5) -- (15,11);
				\draw[lightgray,dashed] (0,4) -- (15,11.5);
			\end{tikzpicture}
	
			\caption{The figure shows the grid for superblock 2, and some of the lines with slope $1/2$. Note that the line passing through $(a, b)$ intersects the $y$-axis at a non-integral point.}
	
			\label{fig:lines}
		\end{figure}

So, all lines of a given superblock has the same slope, and lines from different superblocks have different slopes. As there are $m^{1/6}$ superblocks, and they are numbered $1, 2, \dots, m^{1/6}$, so, we have the slopes of the lines vary as
\begin{equation}
	0 < i \leq m^{1/6}.
	\label{eqn:slope}
\end{equation}

There are two issues to consider -- the number of lines needed to cover every point of the grid, and the purpose of these lines. We address the issue of the count of the lines in this section, and that of the purpose of the lines in the next.

We introduce the notation $l_i(a,b)$ to denote the line that has slope $1/i$, and passes through the point $(a,b)$. We now define the collection of all lines of slope $1/i$ that we are going to draw for the superblock $i$.
\begin{equation}
	L_i = \left\{ \ l_i(a,0) \ \mid \ a \in \mathbb{Z}, \ -i(m^{2/6}-1) \leq a < m^{2/6} \ \right\}.
	\label{eqn:lines}
\end{equation}

In the following three lemmas, we show the properties of this set of lines.

\begin{lemma}
	Every line of $L_i$ contains at least one point of the grid.
	\label{lem:nonempty}
\end{lemma}

\begin{proof}
	Consider an arbitrary line $l_i(a, 0)$ of $L_i$. If $0 \leq a < m^{2/6}$, then $(a, 0)$ itself is a member of the grid, and $l_i(a, 0)$ is non-empty.
	
	Let us now consider the scenario where $-i(m^{2/6} - 1) \leq a < 0$. Let $-a = q i + r$, where $0 \leq r < i$.
	
	If $r = 0$, we show that $(0, q)$ is a point that falls on the line through $(a, 0)$, and it also belongs to the grid. First,
	\[
		\frac{q - 0}{0 - a} = \frac{q}{q i + 0} = \frac{1}{i},
	\]
	which shows that the point falls on the required line. Also,
	\[
		\begin{array}{cccccc}
			&-i(m^{1/2} - 1) & \leq & a & < & 0 \\
			\implies & -i(m^{1/2} - 1) & \leq & - q i - 0 & < & 0 \\
			\implies & m^{1/2} - 1 & \geq & q & > & 0,
		\end{array}
	\]
	which shows that $(0, q)$ belongs to the grid. Together they show that $(0, q) \in l_i(a, 0)$.
	
	On the other hand, if $0 < r < i$, the point to consider is $(i - r, q + 1)$. The following equality shows that the point lies on the line through $(a, 0)$ --
	\[
		\frac{q + 1 - 0}{i - r - a} = \frac{q + 1 - 0}{i - r + q i + r} = \frac{1}{i}.
	\]
	To show that the point belongs to the grid, the $x$-coordinate satisfies the following $0 < i - r < m^{1/6}$ (Equation~\ref{eqn:slope}). As for the $y$-coordinate, we have
	\[
		\begin{array}{cccccc}
			&-i(m^{2/6} - 1) & \leq & a & < & 0 \\
			\implies & -i(m^{2/6} - 1) & \leq & - q i - r & < & 0 \\
			\implies & m^{2/6} - 1 & \geq & q + r / i & > & 0 \\
			\implies & m^{2/6} - 1 & \geq & \lceil q + r / i \rceil & > & 0 \\
			\implies & m^{2/6} - 1 & \geq & q + 1 & > & 0
		\end{array}
	\]
	This shows that even when $r$ in non-zero, $l_i(a, 0)$ is non-empty.
\end{proof}

\begin{lemma}
	Every point of the grid belongs to some line of $L_i$.
	\label{lem:complete}
\end{lemma}

\begin{proof}
	Let $(a, b)$ be an arbitrary element of the grid. By construction, $a$ and $b$ are both integers, and $0 \leq a, b < m^{2/6}$. If $b = 0$, then $(a, 0) \in l_i(a, 0)$.
	
	If $b \neq 0$, consider the point $(a - b i, 0)$. As
	\[
		\frac{a - b i - a}{0 - b} = \frac{1}{i},
	\]
	$(a, b)$ falls on the line through $(a - b i, 0)$. And using arguments similar to the one employed in the previous lemma, one can show that $i (m^{2/6} - 1) \leq a - b i < m^{2/6}$. So, $(a, b)$ falls on the line $l_i(a - b i, 0)$.
\end{proof}

\begin{lemma}
	$\mid L_i \mid \ = \ (i + 1)(m^{2/6} - 1) + 1$.
	\label{lem:lines}
\end{lemma}
\begin{proof}
	The equality is a direct consequence of the definition of $L_i$ (Equation~\ref{eqn:lines}).
\end{proof}

\subsection{Table $\mathcal{B}$}

In table $\mathcal{B}$, we have space to store one block for every line of every superblock. That means that for a superblock, say $i$, all of its blocks that fall on the line $l_i(a,b)$ share the same block in table $\mathcal{B}$; and the same is true for all lines of every superblock.

The $i$\textsuperscript{th} superblock contains $\mid L_i \mid = (i + 1)(m^{2/6} - 1) + 1$ lines (Lemma~\ref{lem:lines}), so the total number of lines from all of the superblocks is
\[
	\begin{array}{rl}
		& \mid L_1 \mid + \mid L_2 \mid + \dots + \mid L_{m^{1/6}} \mid \\
		= & \sum\limits_{i = 1}^{m^{1/6}} \Big( (i + 1)(m^{2/6} - 1) + 1 \Big) \\
		= & \left( \frac{(m^{1/6})(m^{1/6} + 1)}{2} + m^{1/6} \right) (m^{2/6} - 1) + m^{1/6} \\
		= & \mathcal{O}(m^{4/6}).
	\end{array}
\]
As mentioned earlier, we reserve space for one block for each of these lines. Combined with the fact that the size of a block is $m^{1/6}$, we have
\[
	|\mathcal{C}| = \mathcal{O}(m^{5/6}).
\]

\subsection{Notations}

As described in Section~\ref{sec:superblocks}, any element of the universe $\mathcal{U}$ can be addressed by a four-tuple, such as $(s, x, y, i)$, where $s$ is the superblock to which it belongs, $(x,y)$ are the coordinates of its block within that superblock, and $i$ is its index within the block.

Table $\mathcal{A}$ has one bit for each block, so all elements of a block will query the same location. As the block number of the element $(s,x,y,i)$ is $(s,x,y)$, so the bit corresponding to the element is $\mathcal{A}(s, x, y)$; or in other words, the element $(s,x,y,i)$ will query the location $\mathcal{A}(s,x,y)$ in table $\mathcal{A}$.

In table $\mathcal{C}$, there is space for one block for every possible coordinates of the grid. The coordinates of the element $(s,x,y,i)$ is $(x,y)$, and $\mathcal{C}$ has space to store an entire block for this coordinate. So, there is one bit for every element of a block, or, in other words, every index of a block. So, the bit corresponding to the element $(s,x,y,i)$ is $\mathcal{C}(x, y, i)$.

Table $\mathcal{B}$ has a block reserved for every line of every superblock. The element $(s, x, y, i)$ belongs to the line $l_s(x,y)$, and thus table $\mathcal{B}$ has space to store one block corresponding to this line. As the index of the element is $i$, so the bit corresponding to the element in table $\mathcal{B}$ is $\mathcal{B}(l_s(x,y), i)$.

\section{Query Scheme}

The query scheme is easy enough to describe once the data structure has been finalised; it follows the decision tree as discussed earlier (Figure~\ref{fig:tree}). Suppose we want to answer the following membership query -- ``Is $(s, x, y, i)$ in $\mathcal{S}$?'' We would make the first query in table $\mathcal{A}$ at location $\mathcal{A}(s,x,y)$. If the bit stored at that location is 0, we query in table $\mathcal{B}$ at $\mathcal{B}(l_s(x,y), i)$, otherwise we query table $\mathcal{C}$ at $\mathcal{C}(x,y,i)$. If the answer from the second query is 1, then we declare the element to be a member of $\mathcal{S}$, else we declare that it is not a member of $\mathcal{S}$.

\section{The Storage Scheme}

The essence of any bitprobe scheme is the storage scheme, i.e. given a subset $\mathcal{S}$ of the universe $\mathcal{U}$, how the bits of the data structure are set such that the query scheme answers membership questions correctly. We start the description of the storage scheme by giving an intuition for its construction.

\subsection{Intuition}
\label{sec:intuition}

The basic unit of storage in the tables $\mathcal{B}$ and $\mathcal{C}$ of our data structure, in some sense, is a block -- table $\mathcal{B}$ can store one block of any line of any superblock, and table $\mathcal{C}$ can store one block of a given coordinate from any superblock. We show next that our storage scheme must ensure that a empty and non-empty block cannot be stored together in a table.

Suppose, the block $(s,x,y)$ of table $\mathcal{A}$ is non-empty, and it contains the member $(s,x,y,i)$ of subset $\mathcal{S}$. If we decide to store this member in table $\mathcal{B}$, then we have to store the block $(s,x,y)$ in table $\mathcal{B}$. So, we have to set in table $\mathcal{A}$ the following -- $\mathcal{A}(s,x,y)=0$. Thus, $(s,x,y,i)$ upon first query will get a 0 and go to table $\mathcal{B}$. In table $\mathcal{B}$, we store the block $(s,x,y)$ at the storage reserved for the line $l_s(x,y)$. Particularly, we have to set $\mathcal{B}(l_s(x,y),i)=1$.

If $(s,x',y')$ is a block that is empty, i.e. it does not contain any member of $\mathcal{S}$, and it falls on the aforementioned line, i.e. $l_s(x',y')=l_s(x,y)$, then we cannot store this block in table $\mathcal{B}$, and hence $\mathcal{A}(s,x',y')$ must be set to 1. If this is not the case, and $\mathcal{A}(s,x',y')=0$, then the first query for the element $(s,x',y',i)$ will get a 0, go to table $\mathcal{B}$ and query the location $\mathcal{B}(l_s(x',y'),i)$ which is same as $\mathcal{B}(l_s(x,y),i)$. We have set this bit to 1, and we would incorrectly deduce that $(s,x',y',i)$ is a member of $\mathcal{S}$.

The same discussion holds true for table $\mathcal{C}$. If we decide to store the block $(s,x,y)$ in table $\mathcal{C}$, we have to set $\mathcal{A}(s,x,y)$ to 1. In table $\mathcal{C}$, we have space reserved for every possible coordinate for a block, and we would store the block at the coordinate $(x,y)$; particularly, we would set $\mathcal{C}(x,y,i)$ to 1. This implies that all empty blocks from other superblocks having the same coordinate cannot be stored in table $\mathcal{C}$, and hence must necessarily be stored in table $\mathcal{B}$. To take an example, if $(s',x,y)$ is empty, then it must stored it table $\mathcal{B}$, and hence $\mathcal{A}(s',x,y)=0$.

To summarise, for any configuration of the members of subset $\mathcal{S}$, as long as we are able to keep the empty and the non-empty blocks separate, our scheme will work correctly. For the reasons discussed above, we note the following.
\begin{enumerate}
	\item We have to keep the non-empty blocks and empty blocks separate.
	\item We have to keep the non-empty blocks separate from each other; and
	\item The empty blocks can be stored together.
\end{enumerate}
Our entire description of the storage scheme would emphasize on how to achieve the aforementioned objective.

\subsection{Description}

Let the four members of subset $\mathcal{S}$ be
\[
	\mathcal{S} = \Big\{ \ (s_1, x_1, y_1, i_1), \ (s_2, x_2, y_2, i_2), \ (s_3, x_3, y_3, i_3), \ (s_4, x_4, y_4, i_4) \ \Big\}.
\]
So, the relevant blocks are
\[
	\Big\{ \ (s_1, x_1, y_1), \ (s_2, x_2, y_2), \ (s_3, x_3, y_3), \ (s_4, x_4, y_4) \ \Big\},
\]
and the relevant lines are
\[
	\Big\{ \ l_{s_1}(x_1,y_1), \ l_{s_2}(x_2,y_2), \ l_{s_3}(x_3,y_3), \ l_{s_4}(x_4,y_4) \ \Big\}.
\]

In the discussion below, we assume that no two members of $\mathcal{S}$ belong to the same block. This implies that there are exactly four non-empty blocks. The scenario where a block contains multiple members of $\mathcal{S}$ is handled in Section~\ref{sec:multiple}.

The lines for the members of $\mathcal{S}$ need not be distinct, say when two elements belong to the same superblock and fall on the same line. We divide the description of our storage scheme into several cases based on the number of distinct lines we have due to the members of $\mathcal{S}$, and for each of those cases, we provide the proof of correctness alongside it.

\begin{wrapfigure}{r}{0.25\textwidth}
    \centering
    \vspace{-10pt}
	\begin{tikzpicture}[scale=0.4]
		\foreach \i in {0,...,7} {
			\foreach \j in {0,...,7} {
				\draw[gray, fill=gray] (\i,\j) circle (.7ex);
			}
		}
		
		\draw[fill] (1,1) circle (1.5ex);
		\draw[thick,-] (0,0) -- (7,7);
		
		\draw[fill] (3,5) circle (1.5ex);
		\draw[thick,-] (3,0) -- (6.5,7);

		\draw[fill] (4,2) circle (1.5ex);
		\draw[thick,-] (0,4) -- (7,6.33);

		\draw[fill] (6,3) circle (1.5ex);
		\draw[thick,-] (5,0) -- (7,7);
	\end{tikzpicture}
  \end{wrapfigure}

\vspace{-10pt}
\subsubsection{Case I} Suppose we have four distinct lines for the four members of $\mathcal{S}$. The slopes of some of these lines could be same, or they could all be different. We know that all lines of a given superblock have the same slope, and lines from different superblocks have different slopes (Section~\ref{sec:lines}). We also know that if two of these lines, say $l_{s_1}(x_1,y_1)$ and $l_{s_2}(x_2,y_2)$, have the same slope, then the corresponding members of $\mathcal{S}$ belong to the same superblock, i.e. $s_1 = s_2$. On the other hand, if their slopes are distinct, then they belong to different superblocks, and consequently, $s_1 \neq s_2$.

Table $\mathcal{B}$ has space to store one block for every line in every superblock. As the lines for the four members of $\mathcal{S}$ are distinct, the space reserved for the lines are also distinct. So we can store the four non-empty blocks in table $\mathcal{B}$, and all of the empty blocks in table $\mathcal{C}$.

To achieve the objective, we set $\mathcal{A}(s_j,x_j,y_j) = 0$ for $1 \leq j \leq 4$, and set the bits in table $\mathcal{A}$ for every other block to 1. In table $\mathcal{B}$, we set the bits $\mathcal{B}(l_{s_j}(x_j,y_j),i_j) = 1$, for $1 \leq j \leq 4$, and all the rest of the bits to 0. In table $\mathcal{C}$, all the bits are set to 0.

So, if $e$ is an element that belongs to an empty block, it would, according to the assignment above, get a 1 upon its first query in table $\mathcal{A}$. Its second query will be in table $\mathcal{C}$, and as all the bits of table $\mathcal{C}$ are set to 0, we would conclude that the element $e$ is not a member of $\mathcal{S}$.

Suppose, $(s,x,y,i)$ be an element that belongs to one of the non-empty blocks. Then, its coordinates must correspond to one of the four members of $\mathcal{S}$. Without loss of generality let us assume that $s=s_1, x=x_1$, and $y=y_1$.

It follows that $\mathcal{A}(s,x,y)$, which is same as $\mathcal{A}(s_1,x_1,y_1)$, is 0, and hence the second query for this element will be in table $\mathcal{B}$. The line corresponding to the element is $l_s(x,y)$, which is same as $l_{s_1}(x_1,y_1)$, and hence the second query will be at the location $\mathcal{B}(l_s(x,y),i) = \mathcal{B}(l_{s_1}(x_1,y_1),i)$. As the four lines for the four members of $\mathcal{S}$ are distinct, so $\mathcal{B}(l_{s_1}(x_1,y_1),i)$ will be 1 if and only if $i=i_1$. So, we will get a {\tt Yes} answer for your query if and only if the element $(s,x,y,i)$ is actually the element $(s_1,x_1,y_1,i_1)$, a member of $\mathcal{S}$.

\subsubsection{Case II} Let us consider the case when there is just one line for the four members of $\mathcal{S}$. As all of their lines are identical, and consequently, the slopes of the lines are the same, all the elements must belong to the same superblock. So, we have $s_1 = s_2 = s_3 = s_4$.

\begin{wrapfigure}{r}{0.25\textwidth}
    \centering
	\begin{tikzpicture}[scale=0.4]
		\foreach \i in {0,...,7} {
			\foreach \j in {0,...,7} {
				\draw[gray, fill=gray] (\i,\j) circle (.7ex);
			}
		}
		
		\draw[fill] (1,1) circle (1.5ex);
		\draw[fill] (3,3) circle (1.5ex);
		\draw[fill] (5,5) circle (1.5ex);
		\draw[fill] (7,7) circle (1.5ex);
		\draw[thick,-] (0,0) -- (7,7);
	\end{tikzpicture}
  \end{wrapfigure}

As all the non-empty blocks belong to the same superblock, all of their coordinates must be distinct. Table $\mathcal{C}$ can store one block for each distinct coordinate of the grid, and hence we can store the four non-empty blocks there. All the empty blocks will be stored in table $\mathcal{B}$.

To this end, we set $\mathcal{A}(s_j,x_j,y_j) = 1$ for $1 \leq j \leq 4$, and the rest of the bits of table $\mathcal{A}$, which correspond to the empty blocks, to 0. In table $\mathcal{B}$, all bits are set to 0. In table $\mathcal{C}$, the bits corresponding to the four elements are set to 1, i.e. $\mathcal{C}(x_j,y_j,i_j) = 1$ for $1 \leq j \leq 4$. The rest of the bits of table $\mathcal{C}$ are set to 0.

The proof of correctness follows directly from the assignment, and the reasoning follows along the lines of the previous case. If the element $e$ belongs to an empty block, it will get a 0 from table $\mathcal{A}$ upon its first query, consequently go to table $\mathcal{B}$ for its second query, and get a 0, implying $e$ is not a member of $\mathcal{S}$.

If the element $(s,x,y,i)$ belongs to a non-empty block, then its coordinates must correspond to one of the members of $\mathcal{S}$. Without loss of generality, let $s=s_1, x=x_1$, and $y=y_1$.

The first query of the element will be at the location $\mathcal{A}(s,x,y) = \mathcal{A}(s_1,x_1,y_1)$, and hence it will get a 1 from table $\mathcal{A}$, and go to table $\mathcal{C}$ for its second query. In this table, it will query the location $\mathcal{C}(x,y,i)$, which is same as $\mathcal{C}(x_1,y_1,i)$. As the coordinates of the four members of $\mathcal{S}$ are distinct, $\mathcal{C}(x_1,y_1,i)$ will be 1 if and only if $i=i_1$. So, we get a 1 in the second query if and only if we have $(s,x,y,i) = (s_1,x_1,y_1,i_1)$, a member of $\mathcal{S}$.

\subsubsection{Case III} The next case that we consider is when there are two distinct lines corresponding to the four members of subset $\mathcal{S}$. The members can be distributed in one of two ways -- one line contains three elements and the other line one, or the elements might be divided equally among the two lines. We consider the cases separately below.

\paragraph{Case III(A)} Consider the case when one line contains three elements, and the other line contains one. Without loss of generality, let the first three members of $\mathcal{S}$ belong to one line, and the fourth one to another one. So, we have $l_{s_1}(x_1,y_1) = l_{s_2}(x_2,y_2) = l_{s_3}(x_3,y_3)$, and the line $l_{s_4}(x_4,y_4)$ is different from the others. As lines with same slopes belong to the same superblock, we have $s_1 = s_2 = s_3$. Whether the fourth member belongs to the aforementioned superblock, or to a different superblock depends on whether the slope of $l_{s_4}(x_4,y_4)$ is same as the other line or it is distinct.

As the first three elements belong to the same superblock, all will have coordinates distinct from one another. The coordinates of the fourth element could be distinct, or it could overlap with one of the first three.

The case of the coordinates of the four members of $\mathcal{S}$ being distinct is one we have seen in Case II, where the elements too had distinct coordinates. The assignment for this scenario will be identical to that case, and consequently, the correctness proof follows.

\begin{wrapfigure}{r}{0.25\textwidth}
    \centering
	\begin{tikzpicture}[scale=0.4]
		\foreach \i in {0,...,7} {
			\foreach \j in {0,...,7} {
				\draw[gray, fill=gray] (\i,\j) circle (.7ex);
			}
		}
		
		\draw[thick,-] (0,0) -- (7,7);
		\draw[thick,-] (2,0) -- (5.5,7);
		\draw[fill] (0,0) circle (1.5ex);
		\draw[fill=white] (0,1) circle (4ex) node {1};
		\draw[fill] (2,2) circle (1.5ex);
		\draw[fill=white] (2,3) circle (4ex) node {2};
		\draw[fill] (4,4) circle (1.5ex);
		\draw[fill=white] (4,5) circle (4ex) node {3};
		\draw[fill=white] (4,3) circle (4ex) node {4};
	\end{tikzpicture}
\end{wrapfigure}

Let us say that the coordinates of the fourth element coincides with one of the other three members. Without loss of generality, let us assume that the third and the fourth elements have identical coordinates, that is to say $x_3=x_4$ and $y_3=y_4$. As two blocks of a superblock cannot have the same coordinates, we must have $s_3 \neq s_4$. Moreover, different superblocks have different slopes for its lines, implying $l_{s_1}(x_1,y_1) = l_{s_2}(x_2,y_2) = l_{s_3}(x_3,y_3) \neq l_{s_4}(x_4,y_4)$.

The assignment in this case will be as follows -- we will store the blocks corresponding to the first two elements in table $\mathcal{C}$, and the blocks corresponding to the last two elements in table $\mathcal{B}$. The empty blocks accordingly will have to be distributed among the two tables.

Accordingly, we set $\mathcal{A}(s_1,x_1,y_1)$ and $\mathcal{A}(s_2,x_2,y_2)$ to 1, and set $\mathcal{A}(s_3,x_3,y_3)$ and $\mathcal{A}(s_4,x_4,y_4)$ to 0. The bits corresponding to the remaining blocks in the two lines, which are $l_{s_1}(x_1,y_1)$ and $l_{s_4}(x_4,y_4)$, are set to 1. The bits of the blocks of all the other lines in all of the superblocks are set to 0.

In table $\mathcal{B}$, the bits corresponding to the third and the fourth element is set to 1, i.e. $\mathcal{B}(l_{s_3}(x_3,y_3),i_3) = \mathcal{B}(l_{s_4}(x_4,y_4),i_4) = 1$, and all the remaining bits are set to 0. In table $\mathcal{C}$, only the bits corresponding to the first two elements are set to 1, i.e. $\mathcal{C}(x_1,y_1,i_1) = \mathcal{C}(x_2,y_2,i_2) = 1$; the rest of the bits of this table are set to 0.

We now prove that the assignment above is correct. If an element $e$ belongs to a line other than the lines $l_{s_1}(x_1,y_1)$ and $l_{s_4}(x_4,y_4)$, then the bit for its block has been set to 0. Consequently, it will query table $\mathcal{B}$. Table $\mathcal{B}$ has separate space for each line, and only certain bits of the non-empty lines have been set to 1. As $e$ falls on a line different from $l_{s_1}(x_1,y_1)$ and $l_{s_4}(x_4,y_4)$, so the second query for $e$ will also return a 0.

Suppose $e$ belongs to an empty block falling on one the lines $l_{s_1}(x_1,y_1)$ and $l_{s_4}(x_4,y_4)$. According to our assignment, the bits of the empty blocks from the lines are set to 1, and hence the second query for $e$ will go to table $\mathcal{C}$. All blocks falling on a line have distinct coordinates, so the coordinates of the block of $e$ will be distinct from the coordinates of the non-empty blocks of the two lines. As table $\mathcal{C}$ has space to store one block for each distinct coordinate, the space for the empty blocks of the two lines will be different from the non-empty ones. As we have set certain bits of the only the non-empty blocks of table $\mathcal{C}$ to 1, all the bits of the block of $e$ must be 0, and hence the answer to second query for $e$ will be 0.

It remains to verify whether the queries corresponding to the elements of the four non-empty blocks give correct answers. We have argued above that the empty blocks are stored in locations distinct from the non-empty blocks. The assignment tells us that we have stored the non-empty blocks in its entirety. These two facts together imply that queries for elements in the non-empty blocks will also give correct answers.

\paragraph{Case III(B)} We now consider the case when the four members of $\mathcal{S}$ are divided equally among the two lines. Without loss of generality, let us assume that the first two members belong to one line, and the other two members belong to the other line. So, we have $l_{s_1}(x_1,y_1) = l_{s_2}(x_2,y_2)$ and $l_{s_3}(x_3,y_3) = l_{s_4}(x_4,y_4)$. Consequently, we have $s_1 = s_2$ and $s_3 = s_4$.

In this scenario, we may have the four non-empty blocks occupying four distinct coordinates of the grid. This situation is familiar to us, and we will handle it as we have done in Case II.

\begin{wrapfigure}{r}{0.25\textwidth}
    \centering
    \vspace{-20pt}
	\begin{tikzpicture}[scale=0.4]
		\foreach \i in {0,...,7} {
			\foreach \j in {0,...,7} {
				\draw[gray, fill=gray] (\i,\j) circle (.7ex);
			}
		}
		
		\draw[thick,-] (0,0) -- (7,7);
		\draw[thick,-] (2,0) -- (5.5,7);
		\draw[fill] (0,0) circle (1.5ex);
		\draw[fill=white] (0,1) circle (4ex) node {1};
		\draw[fill] (2,0) circle (1.5ex);
		\draw[fill=white] (3,0) circle (4ex) node {3};
		\draw[fill] (4,4) circle (1.5ex);
		\draw[fill=white] (4,5) circle (4ex) node {2};
		\draw[fill=white] (4,3) circle (4ex) node {4};
	\end{tikzpicture}
  \end{wrapfigure}

The other scenario is when coordinates of non-empty blocks overlap. As the lines are distinct, they can have an intersection point if and only if they have different slopes. It implies that the lines belong to different superblocks, and hence $s_1 = s_2 \neq s_3 = s_4$. Further, as there is only one common point between the two lines, only one pair of non-empty blocks from the two lines can overlap, i.e. have the same coordinates. Without loss of generality, let it be the second and fourth member of $\mathcal{S}$. So, we have $x_2 = x_4$ and $y_2 = y_4$.

For all blocks which do not fall on any of the two aforementioned lines, and hence implying that they are empty, we set their bits in table $\mathcal{A}$ to 0. So, the second query for the elements of these blocks will be in table $\mathcal{B}$. As we already know, table $\mathcal{B}$ has seperate space reserved for all lines, and we set all the bits of all of those empty lines to 0.

An important thing to note so far is we have not stored anything in table $\mathcal{C}$ yet. We now look into the assignment of the blocks that fall on the two non-empty lines. The blocks that fall on a line have distinct coordinates, so the blocks on the line $l_{s_1}(x_1,y_1)$ have distinct spaces in table $\mathcal{C}$, and we store all these blocks in table $\mathcal{C}$. We accordingly set the corresponding bits in table $\mathcal{A}$ and $\mathcal{C}$.

We now look into the assignment of the blocks on the other line, namely $l_{s_3}(x_3,y_3)$. There is only one block on this line whose coordinate is same as a point on the other line -- the block corresponding to the fourth member of $\mathcal{S}$ has the same coordinate as the second member of $\mathcal{S}$. Then, we cannot store the block $(s_4,x_4,y_4)$ in table $\mathcal{C}$ as it is already occupied by the block $(s_2,x_2,y_2)$ from the other line. We store this block in table $\mathcal{B}$ at the space reserved for the line $l_{s_3}(x_3,y_3)$. All other blocks of this line can then be stored in table $\mathcal{C}$ without any conflict.

The assignment tells us how the empty and the non-empty blocks have been kept separate. An explicit proof of correctness follows along the lines of the previous cases.

\vspace{-10pt}
\subsubsection{Case IV} The final case to consider is when the number of distinct lines due to the non-empty blocks is three. Without loss of generality, let us assume that the blocks corresponding to the third and fourth elements fall on the same line, i.e. $l_{s_3}(x_3,y_3) = l_{s_4}(x_4,y_4)$. This also means that these two blocks belong to the same superblock, and hence, $s_3 = s_4$. It further implies that the coordinates of the two blocks are distinct.

As seen in the previous cases, those lines of the superblocks which do not contain any non-empty block is easy to handle -- we simply store them in table $\mathcal{B}$ at the space reserved for the respective lines. A point to note is that it also leaves table $\mathcal{C}$ untouched. In the discussion below, we will then concentrate on how we handle the blocks from the three lines which are non-empty.
 
The discussion will be divided into three parts based on how many of those points coincide. As the blocks corresponding to the third and the fourth members have distinct coordinates, it follows that at most three of the non-empty blocks can coincide.

\begin{wrapfigure}{r}{0.25\textwidth}
    \centering
    \vspace{-30pt}
	\begin{tikzpicture}[scale=0.4]
		\foreach \i in {0,...,7} {
			\foreach \j in {0,...,7} {
				\draw[gray, fill=gray] (\i,\j) circle (.7ex);
			}
		}
		
		\draw[thick,-] (0,1) -- (7,4.5);
		\draw[thick,-] (0,0) -- (7,7);
		\draw[thick,-] (1,0) -- (4.5,7);
		\draw[fill] (6,6) circle (1.5ex);
		\draw[fill=white] (6,7) circle (4ex) node {4};
		\draw[fill] (2,2) circle (1.5ex);
		\draw[fill=white] (1,2.5) circle (4ex) node {2};
		\draw[fill=white] (2,1) circle (4ex) node {3};
		\draw[fill=white] (3.5,2) circle (4ex) node {1};
	\end{tikzpicture}
  \end{wrapfigure}

\paragraph{Case IV(A)} Let us consider the scenario when three of the non-empty blocks coincide. Without loss of generality, let it be the first three blocks, i.e. $x_1 = x_2 = x_3$ and $y_1 = y_2 = y_3$.

We store all the blocks on the line $l_{s_3}(x_3,y_3)$ in table $\mathcal{C}$. There is only one point in each of the other two lines, namely $l_{s_1}(x_1,y_1)$ and $l_{s_2}(x_2,y_2)$, that is common with this line -- we store the blocks corresponding to those points in table $\mathcal{B}$, and the rest of the blocks of the other lines in table $\mathcal{C}$. So, the blocks $(s_1,x_1,y_1)$ and $(s_2,x_2,y_2)$ are stored in the location reserved in table $\mathcal{B}$ for the lines $l_{s_1}(x_1,y_1)$ and $l_{s_2}(x_2,y_2)$, and the rest of the blocks of these lines are stored in table $\mathcal{C}$.

This assignment keeps the empty blocks and the non-empty blocks separate from each other, and the correctness follows.

  
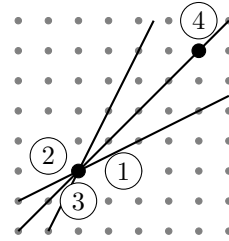
\begin{wrapfigure}{r}{0.25\linewidth}
    \centering
    \vspace{-30pt}
	\begin{tikzpicture}[scale=0.4]
		\foreach \i in {0,...,7} {
			\foreach \j in {0,...,7} {
				\draw[gray, fill=gray] (\i,\j) circle (.7ex);
			}
		}
		
		\draw[thick,-] (0,1) -- (7,1);
		\draw[thick,-] (0,0) -- (7,7);
		\draw[thick,-] (7,1) -- (4,7);
		\draw[fill] (1,1) circle (1.5ex);
		\draw[fill] (7,1) circle (1.5ex);
		\draw[fill=white] (0,1) circle (4ex) node {1};
		\draw[fill=white] (1,0) circle (4ex) node {3};
		\draw[fill=white] (5,2) circle (4ex) node {2};
		\draw[fill=white] (7,0) circle (4ex) node {4};
	\end{tikzpicture}
  \end{wrapfigure}

\paragraph{Case IV(B)} Let us consider the case where two pairs of non-empty blocks coincide. Without loss of generality, let the first block coincide with the third and the second block coincide with the fourth.

The assignment that we devised for the previous case works in this scenario as well -- we store the blocks of the line $l_{s_3}(x_3,y_3)$ in table $\mathcal{C}$, and the blocks $(s_1,x_1,y_1)$ of line $l_{s_1}(x_1,y_1)$ and  $(s_2,x_2,y_2)$ of line $l_{s_2}(x_2,y_2)$ in table $\mathcal{B}$. The other blocks of the lines $l_{s_1}(x_1,y_1)$ and $l_{s_2}(x_2,y_2)$ are stored in table $\mathcal{C}$.

The correctness proof of the previous case holds in this scenario as well.

\paragraph{Case IV(C)} Let us next consider the scenario where only one pair of non-empty blocks coincide. Without loss of generality, let the first block coincide with the third. So, we have $x_1 = x_3$ and $y_1 = y_3$. As only one pair of non-empty blocks coincide, the block of the second element do not lie on any of the other non-empty blocks, and hence has coordinates distinct from the rest.

The assignment in this arrangement will depend on the coordinates of the block of the second block -- it lies on the line $l_{s_3}(x_3,y_3)$, or it doesn't. We address each of these cases below.

  \begin{wrapfigure}{r}{0.25\linewidth}
    \centering
 	\begin{tikzpicture}[scale=0.4]
		\foreach \i in {0,...,7} {
			\foreach \j in {0,...,7} {
				\draw[gray, fill=gray] (\i,\j) circle (.7ex);
			}
		}
		
		\draw[thick,-] (0,1) -- (7,1);
		\draw[thick,-] (0,0) -- (7,7);
		\draw[thick,-] (4.5,0) -- (1,7);
		\draw[fill] (1,1) circle (1.5ex);
		\draw[fill] (7,1) circle (1.5ex);
		\draw[fill] (4,1) circle (1.5ex);
		\draw[fill=white] (0,1) circle (4ex) node {1};
		\draw[fill=white] (1,0) circle (4ex) node {3};
		\draw[fill=white] (5,2) circle (4ex) node {2};
		\draw[fill=white] (7,0) circle (4ex) node {4};
	\end{tikzpicture}
  \end{wrapfigure}

\paragraph{Case IV(C)(i)} We store all of the blocks on the line $l_{s_3}(x_3,y_3)$ in table $\mathcal{C}$. From the line $l_{s_1}(x_1,y_1)$, only one block lies in the previous line, the block containing the first element. This block will be stored in table $\mathcal{B}$ at the location reserved for the line $l_{s_1}(x_1,y_1)$, and the rest of the blocks can be stored in table $\mathcal{C}$. From the last line, i.e. $l_{s_2}(x_2,y_2)$, only one block lies on this line, the block that contains the second element. This blocks will be stored in table $\mathcal{B}$ at the location for the line $l_{s_2}(x_2,y_2)$, and the rest of the blocks can be stored in table $\mathcal{C}$ without conflict.

  
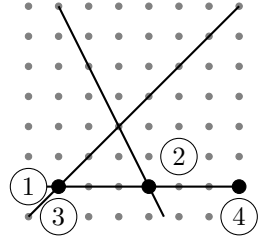
\begin{wrapfigure}{r}{0.25\linewidth}
    \centering
   \vspace{-30pt}
 	\begin{tikzpicture}[scale=0.4]
		\foreach \i in {0,...,7} {
			\foreach \j in {0,...,7} {
				\draw[gray, fill=gray] (\i,\j) circle (.7ex);
			}
		}
		
		\draw[thick,-] (0,1) -- (7,1);
		\draw[thick,-] (0,0) -- (7,7);
		\draw[thick,-] (7,1) -- (4,7);
		\draw[fill] (1,1) circle (1.5ex);
		\draw[fill] (7,1) circle (1.5ex);
		\draw[fill] (5,5) circle (1.5ex);
		\draw[fill=white] (0,1) circle (4ex) node {1};
		\draw[fill=white] (1,0) circle (4ex) node {3};
		\draw[fill=white] (4,5) circle (4ex) node {2};
		\draw[fill=white] (7,0) circle (4ex) node {4};
	\end{tikzpicture}
  \end{wrapfigure}

\paragraph{Case IV(C)(ii)} We next consider the case when the block for the second element does not lie on the line $l_{s_3}(x_3,y_3)$. We, in this case, store the second block, i.e. $(s_2,x_2,y_2)$ in table $\mathcal{C}$, and the rest of the blocks on its line, i.e. $l_{s_2}(x_2,y_2)$, in table $\mathcal{B}$ at its alloted location. We do the same for the block of the first element -- store the non-empty block $(s_1,x_1,y_1)$ in table $\mathcal{C}$, the rest of the blocks on its line $l_{s_1}(x_1,y_1)$ in table $\mathcal{B}$.

The only locations used up in table $\mathcal{C}$ are locations for the first and second block, and the blocks left to be allocated space are those falling on the line $l_{s_3}(x_3,y_3)$. The second block do not lie on this line, and hence would not affect the allocations of the line. The first block coincide with third block falling on this line, so the third block, namely $(s_3,x_3,y_3)$ must necessarily be stored in table $\mathcal{B}$ in the space alloted for the line $l_{s_3}(x_3,y_3)$. The rest of the blocks of the line $l_{s_3}(x_3,y_3)$ can now be stored in table $\mathcal{C}$ without conflict.

\paragraph{Case IV(D)} This is the final configuration to consider when there are three distinct lines due to the non-empty blocks - no block coincide with any other block. This implies that the four non-empty blocks have distinct coordinates, and hence all of them can be stored in table $\mathcal{C}$. All the empty blocks can then be stored in table $\mathcal{B}$, and we would have avoided all conflict.

\subsection{Blocks with Multiple Members}

\label{sec:multiple}

In the discussion above, we had assumed that each block can contain at most one member of the subset $\mathcal{S}$, and we have shown for every configuration of the members of $\mathcal{S}$, the bits of the data structure can be so arranged that the membership queries are answered correctly.

In general, a single block can contain upto four members of $\mathcal{S}$, and we need to propose a assignment for such a scenario. As has been noted in the previous section, our basic unit of storage is a block and we differentiate between empty and non-empty blocks. At a given location in table $\mathcal{B}$ or $\mathcal{C}$, a block is stored in its entirety, or it isn't stored at all. This implies that the number of members of $\mathcal{S}$ a non-empty block contains is of no consequence, as we always store an entire block. The scheme from the previous section would thus hold true for blocks containing multiple members.

We now summarise the result in the theorem below.

\begin{theorem}
	There is an explicit adaptive scheme that stores subsets of size at most four and answers membership queries using two bitprobes such that
	\[
		s_A(4,m,2) = \mathcal{O}(m^{5/6}).
	\]
	\label{thm:final}
\end{theorem}

\vspace{-30pt}
\section{Conclusion}

In this paper, we have proposed an adaptive scheme for storing subsets of size four and anwering membership queries with two bitprobes that improves upon the existing schemes in the literature. The technique used is that of arranging the blocks of a superblock in a two-dimenstional grid, and grouping them along lines. We hope that this technique can be extended to store larger subsets by extending the idea of an arrangement in a two-dimensional grid to arrangements in three and higher dimensional grids.

\bibliographystyle{splncs04}
\bibliography{main}

\begin{thebibliography}{1}
\providecommand{\url}[1]{\texttt{#1}}
\providecommand{\urlprefix}{URL }
\providecommand{\doi}[1]{https://doi.org/#1}

\bibitem{DBLP:conf/soda/AlonF09}
Alon, N., Feige, U.: On the power of two, three and four probes. In:
  Proceedings of the Twentieth Annual {ACM-SIAM} Symposium on Discrete
  Algorithms, {SODA} 2009, New York, NY, USA, January 4-6, 2009. pp. 346--354
  (2009)

\bibitem{DBLP:conf/walcom/BaigK18}
Baig, M.G.A.H., Kesh, D.: Two new schemes in the bitprobe model. In: {WALCOM:}
  Algorithms and Computation - 12th International Conference, {WALCOM} 2018,
  Dhaka, Bangladesh, March 3-5, 2018, Proceedings. pp. 68--79 (2018)

\bibitem{garg15}
Garg, M.: The Bit-probe Complexity of Set Membership. Ph.D. thesis, School of
  Technology and Computer Science, Tata Institute of Fundamental Research, Homi
  Bhabha Road, Navy Nagar, Colaba, Mumbai 400005, India (2016)

\bibitem{DBLP:conf/soda/GargR15}
Garg, M., Radhakrishnan, J.: Set membership with a few bit probes. In:
  Proceedings of the Twenty-Sixth Annual {ACM-SIAM} Symposium on Discrete
  Algorithms, {SODA} 2015, San Diego, CA, USA, January 4-6, 2015. pp. 776--784
  (2015)

\bibitem{kesh32}
Kesh, D.: Space complexity of two adaptive bitprobe schemes storing three
  elements (2018), http://www.iitg.ac.in/deepkesh/paper32.pdf

\bibitem{DBLP:conf/esa/LewensteinMNR14}
Lewenstein, M., Munro, J.I., Nicholson, P.K., Raman, V.: Improved explicit data
  structures in the bitprobe model. In: Algorithms - {ESA} 2014 - 22th Annual
  European Symposium, Wroclaw, Poland, September 8-10, 2014. Proceedings. pp.
  630--641 (2014)

\bibitem{DBLP:conf/esa/RadhakrishnanRR01}
Radhakrishnan, J., Raman, V., Rao, S.S.: Explicit deterministic constructions
  for membership in the bitprobe model. In: Algorithms - {ESA} 2001, 9th Annual
  European Symposium, Aarhus, Denmark, August 28-31, 2001, Proceedings. pp.
  290--299 (2001)

\bibitem{DBLP:conf/esa/RadhakrishnanSS10}
Radhakrishnan, J., Shah, S., Shannigrahi, S.: Data structures for storing small
  sets in the bitprobe model. In: Algorithms - {ESA} 2010, 18th Annual European
  Symposium, Liverpool, UK, September 6-8, 2010. Proceedings, Part {II}. pp.
  159--170 (2010)

\end{thebibliography}

\end{document}